\documentclass[10pt,twocolumn,aps,pra,amsmath,amssymb,showpacs]{revtex4-1}
\usepackage{bm}
\usepackage{mathrsfs}
\usepackage{graphicx}
\usepackage[usenames,dvipsnames,svgnames,table]{xcolor}
\usepackage[unicode=true,
            pdfusetitle,
            bookmarks=true,
            bookmarksnumbered=false,
            bookmarksopen=false,
            breaklinks=true,
            pdfborder={0 0 0},
            backref=false,
            colorlinks=true]{hyperref}
\hypersetup{linkcolor=NavyBlue,urlcolor=NavyBlue,citecolor=NavyBlue}

\usepackage{amsthm}
\newtheorem{prop}{\protect\propositionname}
\providecommand{\propositionname}{Proposition}

\usepackage{tikz-cd}

\DeclareMathOperator{\tr}{Tr} 			
\DeclareMathOperator{\id}{\mathrm{id}} 	
\newcommand{\idop}{\openone} 			

\newcommand{\cB}{\mathcal{B}}
\newcommand{\cD}{\mathcal{D}}
\newcommand{\cH}{\mathcal{H}}

\begin{document}

\title{Structure of correlated initial states that guarantee completely positive reduced dynamics}
\author{Xiao-Ming Lu}
\email{luxiaoming@gmail.com}
\affiliation{Department of Electrical and Computer Engineering, National University of Singapore, 4 Engineering Drive 3, Singapore 117583}

\begin{abstract}
We use the Koashi-Imoto decomposition of the degrees of freedom of joint system-environment initial states to investigate the reduced dynamics.
We show that a subset of joint system-environment initial states guarantees completely positive reduced dynamics, if and only if the system privately owns all quantum degrees of freedom and can locally access the classical degrees of freedom, without disturbing all joint initial states in the given subset.
Furthermore, we show that the quantum mutual information for such kinds of states must be independent of the quantum degrees of freedom. 
\end{abstract}


\maketitle

\section{Introduction}
As the standard tool mathematically describing quantum operations and quantum channels, the completely positive (CP) linear maps between quantum states are of great importance in quantum information processing and the dynamics of open quantum systems~\cite{Kraus1983,Nielsen2000,Choi1975,Davies1976,Alicki2007,Breuer2002,Rivas2012}. 
It is well known that for uncorrelated joint initial states with a fixed environmental state, the reduced dynamical maps for open quantum systems from the initial time to any later time are always CP, no matter what the interaction is~\cite{Davies1976,Alicki2007,Breuer2002,Rivas2012}.
For correlated joint initial states, characterizing the reduced dynamical maps is much more complicated~\cite{Pechukas1994,Alicki1995,Pechukas1995,Stelmachovic2001,Jordan2004,Carteret2008,Dai2015}, and until now not too much about it has been known.

However, restricting to some special subsets of (correlated) joint initial states, the reduced dynamical maps can be described by CP linear maps~\cite{Pechukas1994,Alicki1995}, so that many elegant structures and features of CP linear maps can be directly applied to study the reduced dynamics, e.g., the non-Markovianity based on indivisibility of CP dynamical maps~\cite{Wolf2008,Breuer2009,Rivas2010,Lu2010a}.
Some efforts have been devoted to pursue a complete characterization of the structure of subsets of joint initial states that guarantee CP reduced dynamics for any joint unitary evolution~\cite{Rodriguez-Rosario2008,Shabani2009,Brodutch2013,Liu2014,Buscemi2014a,Dominy2016,Dominy2016a,Tuerkmen2016}.
It is particularly concerned whether the system-environment correlations are relevant or not~\cite{Rodriguez-Rosario2008,Shabani2009,Brodutch2013,Liu2014,Dominy2016,Tuerkmen2016}, since joint initial states that have only classical correlations was shown to lead to CP reduced dynamics~\cite{Rodriguez-Rosario2008}.
Recently, it was shown that the purely classical correlations is not necessary for the CP reduced dynamics, and in principle any kind of correlation can be present in the joint initial states guaranteeing CP reduced dynamics~\cite{Brodutch2013,Liu2014}. 
An alternative approach is to consider the families of joint initial states that are post-selected via an ancillary from an ancillary-system-environment tripartite state.
In such a case, the reduced dynamics is CP if and only if the tripartite state is a short quantum Markov chain~\cite{Buscemi2014a}.
The special case where the tripartite short quantum Markov chain is a pure state was elaborately discussed in Ref.~\cite{Tuerkmen2016}. 
Yet despite all that, a structure theory of subsets of joint initial states guaranteeing CP reduced dynamics is still not available.

In this paper, we use the Koashi-Imoto (KI) decomposition to tackle this question. 
The KI decomposition is an elegant way to characterize the structure of subsets of quantum states~\cite{Koashi2002}.
In terms of how the information carried by individual states can be accessed by the operations that leave invariant all states in the given subset, the KI decomposition classifies the degrees of freedom of individual states into three parts:
the quantum part is inaccessible, the classical part is read-only, and the redundant part is the same for all states in the given subset~\cite{Koashi2002}. 
We show that a subset of joint initial states guarantees CP reduced dynamics, if and only if the system privately owns all quantum degrees of freedom and can locally access the classical degrees of freedom, without disturbing all joint initial states in the given subset.  
Moreover, we show that the correlations play a role in the CP reduced dynamics from a structural perspective:
the quantum mutual information must be independent of the quantum degrees of freedom.

This paper is organized as follows. 
In Sec.~\ref{sec:reduced_dynamics}, we give a brief review on the reduced dynamics and the issue of the complete positivity of reduced dynamical maps. 
In Sec.~\ref{sec:KI}, we introduce the KI decomposition for a subset of quantum states.
In Sec.~\ref{sec:structure_theory}, we use the KI decomposition to investigate the structure of the joint initial states that guarantee the complete positivity of reduced dynamical maps.
We summarize our results in Sec.~\ref{sec:conclusion}.

\section{Reduced dynamics}\label{sec:reduced_dynamics}

Let $\cH_s$ and $\cH_e$ be the Hilbert spaces of the system and its environment respectively, which are assumed to be finite-dimensional throughout this work. 
The entirety as a closed dynamical system experiences a joint unitary evolution as $\rho_{se}\mapsto U\rho_{se} U^{\dagger}$, where $U$ is a unitary operator acting on $\cH_s\otimes\cH_e$.
The reduced dynamical maps are defined as the transformation of the reduced states of the system, namely, 
\begin{equation}\label{eq:rd}
	\Upsilon:\rho_s\mapsto\tr_e(U\rho_{se} U^\dagger),
\end{equation}
where $\rho_s=\tr_e(\rho_{se})$ is the reduced density operator for the open quantum system, and $\tr_e$ denotes the partial trace with respect to the environment.
To make $\Upsilon$ well-defined by Eq.~\eqref{eq:rd} for any unitary transformation $U$, the joint initial states $\rho_{se}$ must be restricted in a subset, known as the {\em initial condition}, such that each $\rho_s$ corresponds to a single $\rho_{se}$.
Otherwise, for some unitary transformations there will be two different final reduced states that correspond to a single reduced initial state~\cite{Pechukas1994,Dominy2016,Dominy2016a}. 
In other words, the partial trace $\tr_e$ must be a one-to-one map, when restricted to the subset of joint initial states.
The reduced dynamical map can be generalized by identifying a different system-environment partition at the final time~\cite{Buscemi2014a}, namely, 
\begin{equation}
	\Upsilon[s\to s']:\rho_s \to \tr_{e'}(U\rho_{se} U^\dagger),
\end{equation}
where $s'$ and $e'$ constitute another system-environment partition such that $\cH_s\otimes\cH_e=\cH_{s'}\otimes\cH_{e'}$.

Let $\cB(\cH)$ be the space of all bounded linear operators on a Hilbert space $\cH$. 
A linear map $\Phi$ from $\cB(\cH_1)$ to $\cB(\cH_2)$ is called a CP map, if $\id_n\otimes\Phi$ maps all positive semidefinite operator in $\cB(\mathbb{C}^n\otimes\cH_1)$ to positive semidefinite operators in $\cB(\mathbb{C}^n\otimes\cH_2)$ for any positive integer $n$, where $\id_n$ denote the identity maps from $\cB(\mathbb{C}^n)$ to $\cB(\mathbb{C}^n)$.
We say that a reduced dynamical map $\Upsilon$ admits a CP description, if there exists a CP extension of $\Upsilon$ to all density operators of the system, that is, there exists a CP map $\Phi$ such that $\Upsilon(\rho_s)=\Phi(\rho_s)$ for all $\rho_s$ in the domain of $\Upsilon$.
Furthermore, we say that a subset of joint initial states {\em guarantees CP reduced dynamics}, if any generalized reduced dynamical map $\Upsilon[s\to s^\prime]$ (for any joint unitary transformation and any system-environment partition at the final time) admits a CP description.
Characterizing the structure of the subsets of joint initial states that guarantee CP reduced dynamics is the main purpose of this work.  

\section{KI decomposition}\label{sec:KI}

The CP description of reduced dynamics, if exists, is in general dependent on the common characteristics shared by the joint initial states, and independent of the individual characteristics carried by the joint initial states in the given subset.
For instance, in the canonical case of product initial states, the Kraus operators of the reduced dynamical map depend on the fixed state of the environment.
Similarly, for a set of correlated initial states, the information concerning the common characteristics of joint initial states can be incorporated into the description of the reduced dynamical map.
The KI decomposition~\cite{Koashi2002} is such a method that separates the common characteristics and the individual characteristics of states in a given subset, and thus will be used for this problem.

The KI decomposition is restated by Hayden {\it et al.}~\cite{Hayden2004} as follows.
For a set $\cD$ of finite-dimensional density operators, there exists a unique decomposition of the Hilbert space as 
\begin{equation}\label{eq:KI_Hilbert_Space}
	\cH = \bigoplus_j \cH_{l_j}\otimes\cH_{r_j}
\end{equation}
such that the following two conditions are satisfied.
(i) Every density operator $\rho$ in $\cD$ can be expressed as 
\begin{equation}\label{eq:KI_states}
	\rho = \bigoplus_j p_j \rho_{l_j}\otimes\tilde\omega_{r_j}.
\end{equation}
(ii) Every quantum operation $\Phi$ that leaves all $\rho$ in $\cD$ invariant satisfies 
\begin{equation}\label{eq:KI_operation}
	\Phi|_{\cB(\cH_{l_j} \otimes \cH_{r_j})} = \id_{l_j} \otimes \Phi_{r_j},
\end{equation}
where $\{p_j\}$ is a probability distribution, $\tilde\omega_{r_j}$ are fixed density operators on $\cH_{r_j}$, $\id_{l_j}$ are the identity maps on $\cB(\cH_{l_j})$, and $\Phi_{r_j}$ are quantum operations---CP and trace-preserving linear maps---satisfying $\Phi_{r_j}(\tilde\omega_{r_j})=\tilde\omega_{r_j}$ for all $j$. 
In Eq.~\eqref{eq:KI_states}, the direct sum $\oplus$ of matrices $A_j$ stands for
\begin{equation}
\bigoplus_{j=1}^{n} A_j = 
\begin{bmatrix}
 A_1 & \mathbf{0} & \cdots & \mathbf{0} \\
 \mathbf{0} & A_2 & \cdots & \mathbf{0} \\
 \vdots & \vdots & \ddots & \vdots \\
 \mathbf{0} & \mathbf{0} & \cdots & A_n \\
\end{bmatrix},
\end{equation}
where $\mathbf{0}$ is a zero matrix of proper dimension.
The tilde notation, $\tilde\bullet$, on a density operator means that we are considering a fixed density operator, henceforth.

By the KI decomposition, the degrees of the freedom of states in a given subset are classified into three parts in terms of Eq.~(\ref{eq:KI_states}), see Ref.~\cite{Koashi2002}: 
The classical part $\{p_j\}$ can be accessed by the projective measurement, whose elements are the projections onto the subspace $\cH_{l_j}\otimes\cH_{r_j}$, without disturbing all states in the given subset. 
The quantum part $\{\rho_j\}$ are inaccessible by the quantum operations that do not disturb all states in the given subset. 
The knowledge about the Hilbert space decomposition Eq.~\eqref{eq:KI_Hilbert_Space} and the redundant part $\{\tilde \omega_j\}$ are the same for all states in the given subset, thus can be incorporated into the CP description of reduced dynamical maps.

\section{Structure of joint initial states}\label{sec:structure_theory}

\begin{figure}
\includegraphics{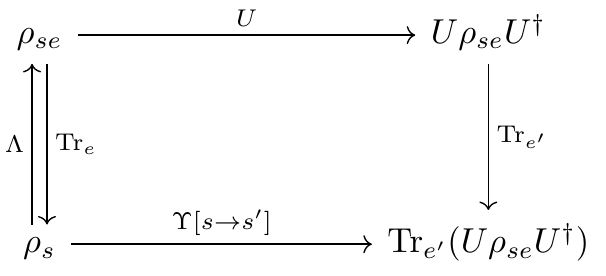}
\caption{\label{fig:fig1} Commutative diagram of the relevant maps between quantum states.} 
\end{figure}

In order to bridge the complete positivity of reduced dynamical maps and the KI decompositions, we resort to the assignment maps~\cite{Pechukas1994}, which assigns each reduced initial state in the domain to a joint initial state. 
See Fig.~\ref{fig:fig1} for the relationships between the relevant maps. 
Note that an assignment map $\Lambda$ can be in principle defined on the domain as the inverse of $\tr_e$, as the joint initial states $\rho_{se}$ must be restricted in a subset such that each $\rho_s$ corresponds to a single $\rho_{se}$, in order to make the reduced dynamical maps given by Eq.~\eqref{eq:rd} well-defined for any joint unitary transformation $U$. 
The generalized reduced dynamical maps can then be expressed as $\Upsilon[s\to s^\prime]:\rho_s \mapsto \tr_{e^\prime}(U\Lambda(\rho_s)U^\dagger)$. 

Now, we give our main result by the following proposition.

\begin{prop}\label{prop:main}
For a given subset $\cD_{se}$ of joint initial states, the following statements are equivalent.
\begin{enumerate}
	\item[{\rm (a)}] Any generalized reduced dynamical map admits a CP description.
	\item[{\rm (b)}] There exists a CP assignment map, in the sense that the assignment map $\Lambda$ defined on the subset of reduced initial states admits a CP extension to all density operators.
	\item[{\rm (c)}] There exists a decomposition $\cH_s=\bigoplus_j\cH_{l_j}\otimes\cH_{r_j}$ for the system Hilbert space such that every state in $\cD_{se}$ can be expressed as 
\begin{equation}\label{eq:state_decomposition}
	\rho_{se} = \bigoplus_j p_j \rho_{l_j}\otimes\tilde\omega_{r_j e},
\end{equation}
where $\{p_j\}$ is a probability distribution, $\rho_{l_j}$ are density operators on $\cH_{l_j}$, and $\tilde\omega_{r_j e}$ are fixed density operators on $\cH_{r_j}\otimes\cH_e$. 
\end{enumerate}	
\end{prop}

\begin{proof}
The fact that (a) implies (b) can be seen by noting that the generalized reduced dynamical map with $U=\idop_s\otimes\idop_e$ and $\cH_{s^\prime}=\cH_s\otimes\cH_e$ itself is an assignment map, where $\idop_s$ and $\idop_e$ are the identity operators on $\cH_s$ and $\cH_e$ respectively.

We now prove that (b) implies (c).
Let $\Lambda$ be a CP assignment map satisfying $\tr_e[\Lambda(\rho_s)]=\rho_s$ for all $\rho_s$ in the domain $\cD_s:=\{\tr_e\rho_{se}|\rho_{se}\in\cD_{se}\}$.
Let $\cH_s=\bigoplus_j\cH_{l_j}\otimes\cH_{r_j}$ be the KI decomposition of the system Hilbert space for $\cD_s$.
Since the composite map $\tr_e\circ \Lambda$ leaves all density operators in $\cD_s$ invariant, it must satisfy Eq.~\eqref{eq:KI_operation} with $\Phi=\tr_e\circ \Lambda$, that is,
\begin{equation}
	\tr_e\circ \Lambda|_{\cB(\cH_{l_j}\otimes\cH_{r_j})}=\id_{l_j}\otimes\tr_e\circ\Lambda_{r_j}
\end{equation}
with $\Lambda_{r_j}$ being a CP map from $\cB(\cH_{r_j})$ to $\cB(\cH_{r_j}\otimes \cH_e)$ such that $\tr_e\circ\Lambda_{r_j}(\tilde\omega_{r_j})=\tilde\omega_{r_j}$.
Thus, $\rho_{se}=\Lambda(\rho_s) = \bigoplus_j p_j \rho_{l_j} \otimes \tilde\omega_{r_j e}$ with $\tilde\omega_{r_j e}= \Lambda_{r_j}(\tilde\omega_{r_j})$, which is in the form of Eq.~\eqref{eq:state_decomposition}.

To prove that (c) implies (a), we explicitly construct a CP assignment map for the states of the form Eq.~\eqref{eq:state_decomposition} as follows:
\begin{equation}
	\Lambda(\rho_s)=\sum_j\tr_{r_j}(\Pi_j\rho_s\Pi_j)\otimes\tilde\omega_{r_j e},
\end{equation}
where $\Pi_j$ are the projections from $\cH_s$ onto the subspace $\cH_{l_i}\otimes\cH_{r_j}$.
Note that the unitary channel and the partial trace are both CP, and the compositions of CP maps are also CP. 
Therefore, any generalized reduced dynamical map, $\Upsilon[s\to s^\prime]:\rho_s \mapsto \tr_{e^\prime}(U\Lambda(\rho_s)U^\dagger)$, is CP.
\end{proof}

We shall show how our result are related to the previous ones. 
First, if we require the reduced initial states to contain all density operators on the system Hilbert space, then the KI decomposition of the reduced initial states can only have a single subspace and $\cH_{l_1}=\cH_s$. 
In other words, the joint initial states must be product states with a fixed environmental state, which recovers Pechukas's result~\cite{Pechukas1994}.
Second, Liu and Tong~\cite{Liu2014} showed that the CP reduced dynamics is guaranteed if the joint initial states are in the subsets of the form 
\begin{equation}\label{eq:Liu_Tong}
	\cD_{se}=\left\{
	\left(\bigoplus_{j=1}^{n}p_j\tilde\omega_{se,j}\right)
	\oplus
	\left(\bigoplus_{j=n+1}^Np_j\rho_{s,j}\otimes\tilde\omega_{e,j}\right)
	\right\},
\end{equation} 
where the system Hilbert space is decomposed into $\cH_s=\bigoplus_{j=1}^N\cH_{s,j}$, $\{p_j\}$ is an arbitrary probability distribution, $\rho_{s,j}$ are arbitrary density operators on $\cH_{s,j}$, $\tilde\omega_{se,j}$ are fixed density operator on $\cH_{s,j}\otimes\cH_e$, and $\tilde\omega_{e,j}$ are fixed operators on $\cH_e$. 
This form includes as special cases the previous examples of the subset of joint initial states with vanishing quantum discord~\cite{Rodriguez-Rosario2008} as well as with non-vanishing quantum discord~\cite{Brodutch2013}. 
It is easy to see that the states given in Eq.~\eqref{eq:Liu_Tong} are in the form of Eq.~\eqref{eq:state_decomposition} with $\cH_{l_j}$ being either 1-dimensional or of the same dimension as $\cH_{s,j}$.

Third, Buscemi considered a typical kind of joint initial states that are generated from a fixed ancillary-system-environment tripartite state by post-selecting in the ancillary, that is,
\begin{equation}\label{eq:post_select}
	\cD_{se}=\left\{
	\frac{\tr_a[(E_a\otimes \openone_s\otimes\idop_e)\tilde\rho_{ase}]}
	{\tr[(E_a\otimes \openone_s\otimes\idop_e)\tilde\rho_{ase}]} \mid 0\leq E_a \leq \openone_a
	\right\}
\end{equation}
with $\idop_a$ being the identity operator on the ancillary Hilbert space~\cite{Buscemi2014a}.
It was shown that a post-selected family of joint initial states guarantee CP reduced dynamics, if and only if the tripartite state $\tilde\rho_{ase}$ is an ancillary-system-environment short quantum Markov chain~\cite{Buscemi2014a}.
Note that a short quantum Markov chain $\tilde\rho_{ase}$ can be expressed as
\begin{equation}
	\tilde\rho_{ase}=\bigoplus_j \tilde p_j \tilde\rho_{al_j}\otimes\tilde\rho_{r_je}
\end{equation}
with the decomposition $\cH_s=\bigoplus_j \cH_{l_j}\otimes\cH_{r_j}$ for the system Hilbert space, where $\{\tilde p_j\}$ is a probability distribution, $\tilde\rho_{al_j}$ and $\tilde\rho_{r_je}$ are density operators on $\cH_a\otimes\cH_{l_j}$ and $\cH_{r_j}\otimes\cH_e$ respectively (see Theorem 6 in Ref.~\cite{Hayden2004}).
Then, it is easy to see that the post-selected family from a tripartite short quantum Markov chain are in the form of Eq.~\eqref{eq:state_decomposition}.

In fact, the post-selected families used in Ref.~\cite{Buscemi2014a} are sufficiently general such that a family of joint states in the form of Eq.~\eqref{eq:state_decomposition} can always be generated as a post-selected family, Eq.~\eqref{eq:post_select}, with a tripartite short quantum Markov chain~\cite{Buscemi_private_communication}
\begin{equation}
	\tilde\rho_{ase}=\sum_j|j\rangle\langle j|_{a_1}\otimes\phi_{a_2,l_j}\otimes\tilde\omega_{r_je},
\end{equation} 
where the ancillary is associated with the Hilbert space $\cH_{a_1}\otimes\cH_{a_2}$, and  $\phi_{a_2,l_j}=\sum_{k,l=1}^{\dim\cH_{l_j}} |k\rangle\langle l|\otimes|k\rangle\langle l|/(\dim\cH_{l_j})$ are maximally entangled states defined on $\cH_{a_2}\otimes\cH_{l_j}$ for all $j$.  
Note that the dimension of $\cH_{a_2}$ should be not less than the maximal dimension of $\cH_{l_j}$ for all $j$.
Thus, any subset of joint initial states that guarantee the CP reduced dynamics can be generated as a subset of a post-selected family from a specific tripartite short quantum Markov chain.  
The equivalence between the families given by Eq.~\eqref{eq:state_decomposition} and Eq.~\eqref{eq:post_select} is also proved in Ref.~\cite{Yu}

From the perspective of the classification of degrees of freedom by the KI decomposition, Proposition~\ref{prop:main} means that a subset of joint initial states guarantees CP reduced dynamics, if and only if all quantum degrees of freedom are  privately owned by the system, and the classical degrees of freedom can be locally accessed by the system without disturbance.
Here, the local accessibility of the classical degrees of freedom without disturbance means that there exists a local measurement, performed only on the system, such that the classical probability distribution can be produced without disturbing all joint initial states in the given subset. 
To be precise, such a local measurement is the projective measurement with respect to the projections from $\cH_s$ onto $\cH_{l_j}\otimes\cH_{r_j}$.

It is of interest whether correlation between the system and its environment plays an essential role in the complete positivity of the reduced dynamics restricted to a given subset of joint initial states. 
It was pointed out in Ref.~\cite{Brodutch2013,Liu2014} that there is no definite relation between correlations in the individual states and the complete positivity of the reduced dynamics;
this also can be seen by noting that the redundant part, the fixed density operators $\{\tilde\omega_{r_je}\}$ in Eq.~\eqref{eq:state_decomposition}, may in principle possess any kind of correlation.
Here, we show that there is indeed a relation between the CP reduced dynamics and the classicality of the correlation in another sense: 
the quantum mutual information of the joint initial states that guarantee CP reduced dynamics is independent of the quantum degrees of freedom.
This can be seen by substituting Eq.~\eqref{eq:state_decomposition} into the quantum mutual information $I(\rho_{se})=S(\rho_s)+S(\rho_e)-S(\rho_{se})$, where $S(\rho)=-\tr(\rho\ln\rho)$ is the von Neumann entropy.
It follows after some algebras that
\begin{align}\label{eq:mutual_information}
	I(\rho_{se}) &= S(\rho_e) + \sum_j p_j \left(S(\tilde\omega_{r_j}) - S(\tilde\omega_{r_je})\right),
\end{align}
where $\tilde\omega_{r_j}:=\tr_e(\tilde\omega_{r_je})$ and $\rho_e:=\tr_s\rho_{se}=\sum_j p_j \tr_{r_j}\tilde\omega_{r_je}$.
In the derivation of Eq.~\eqref{eq:mutual_information}, we have used $S(\oplus_jp_j\rho_j)=-\sum_jp_j\ln p_j+\sum_jp_jS(\rho_j)$ and $S(\rho\otimes\sigma)=S(\rho)+S(\sigma)$.

We here give the Kraus operators of the CP reduced dynamical maps for a subset of joint initial states given by Eq.~\eqref{eq:state_decomposition} and a given joint unitary operator $U$. 
For simplicity, it is enough to consider the union of the supports of all reduced initial states, and from now on we will assume that this reduction of the system Hilbert space has been done.
As a consequence, $\tilde\omega_{r_j}$, defined as $\tr_e\tilde\omega_{r_je}$, have full ranks on the Hilbert space $\cH_{r_j}$.
Then, generalized reduced dynamical maps can be expressed as $\Upsilon[s\to s^\prime](\rho_s) = \sum_{\alpha\beta} K_{\alpha\beta} \rho_s K_{\alpha\beta}^\dagger$ with the Kraus operators from $\cH_s$ to $\cH_{s^\prime}$ being defined by $K_{\alpha\beta} = \langle\alpha_{e^\prime}|UA|\beta_e\rangle$, 
where $A:=\sum_j\Pi_j\tilde\omega_{r_je}^{1/2}\tilde\omega_{r_j}^{-1/2}$ is an operator on $\cH_s\otimes\cH_e$, and $\{|\alpha_{e^\prime}\rangle\}$ and $\{|\beta_e\rangle\}$ are bases of $\cH_{e^\prime}$ and $\cH_e$ respectively.
It is easy to convince oneself that $K_{\alpha\beta}$ is a set of Kraus operators by noting that
\begin{multline}
	\sum_{\alpha\beta} K_{\alpha\beta}^\dagger K_{\alpha\beta} = \tr_e(A^\dagger A) \\
	= \sum_{jk} \Pi_j\tilde\omega_{r_j}^{-1/2}\tr_e(\tilde\omega_{r_je}^{1/2}\tilde\omega_{r_ke}^{1/2})\tilde\omega_{r_k}^{-1/2}\Pi_k
	= \idop_s.
\end{multline}
In the derivation of the third equality above, we have used $\tilde\omega_{r_je}^{1/2}\tilde\omega_{r_ke}^{1/2}=\tilde\omega_{r_je}^{1/2}\delta_{jk}$, as the density operators $\tilde\omega_{r_je}$ and $\tilde\omega_{r_ke}$ are in different subspaces for $j\neq k$.

\section{Conclusion}\label{sec:conclusion}
In summary, we have investigated the complete positivity of the reduced dynamics, from the structural perspective of considering the degrees of freedom of joint initial states in a given subset.
We have shown that a subset of joint system-environment initial states guarantees CP reduced dynamics, if and only if the system privately owns all quantum degrees of freedom and can locally access the classical degrees of freedom, without disturbing all joint initial states in the given subset.
Moreover, we have discussed, also from the structural perspective, the relation between the CP reduced dynamics and the classicality of the system-environment correlations.

\section*{Acknowledgment}
Discussions with F. Buscemi and D. M. Tong are gratefully acknowledged.
This work is supported by the Singapore National Research Foundation under NRF Grant No.~NRF-NRFF2011-07 and the Singapore Ministry of Education Academic Research Fund Tier 1 Project R-263-000-C06-112.

\bibliography{research}

\onecolumngrid

\end{document}